\documentclass[lettersize,journal]{IEEEtran}
\usepackage{amsmath,amsfonts}
\usepackage{array}
\usepackage{textcomp}
\usepackage{stfloats}
\usepackage{url}
\usepackage{verbatim}
\usepackage{cite}
\usepackage{cite}
\usepackage{amssymb}
\usepackage{amsfonts}
\usepackage{amsmath}
\usepackage{amsthm}
\usepackage{mathrsfs}
\usepackage{bm}
\usepackage{verbatim}
\usepackage{autobreak} 
\usepackage{lettrine}
\usepackage{graphicx} 
\usepackage{float} 
\usepackage{subfigure} 
\usepackage{tabularx}
\usepackage{algpseudocode} 
\usepackage{enumerate}
\usepackage{cases}
\usepackage[linesnumbered, ruled]{algorithm2e}

\usepackage{url}
\usepackage[colorlinks,
            linkcolor=black,
            anchorcolor=black,
            citecolor=blue]{hyperref}
\usepackage{hyperref}
\newtheorem{lemma}{Lemma}

\newtheorem{remark}{Corollary}
\newtheorem{proposition}{Proposition}

\allowdisplaybreaks[3]
\usepackage[monochrome]{color}

\begin{document}


\title{Robust Beamforming Design and Antenna Selection for Dynamic HRIS-aided MISO Systems}

\author{ Jintao Wang, Binggui Zhou, Chengzhi Ma, Shiqi Gong, Guanghua Yang, Shaodan Ma
\thanks{Jintao Wang, Binggui Zhou, Chengzhi Ma, and Shaodan Ma are with the State Key Laboratory of Internet of Things for Smart City and the Department of Electrical and Computer Engineering, University of Macau, Macau 999078, China (e-mail: wang.jintao@connect.um.edu.mo; binggui.zhou@connect.um.edu.mo; yc07499@um.edu.mo; shaodanma@um.edu.mo).}


\thanks{Shiqi Gong is with the School of Cyberspace Science and Technology, Beijing Institute of Technology, Beijing 100081, China (e-mail:gsqyx@163.com).}
\thanks{Guanghua Yang is with the School of Intelligent Systems Science and Engineering, Jinan University, Zhuhai 519070, China (e-mail: ghyang@jnu.edu.cn).}
}

\markboth{Journal of \LaTeX\ Class Files, September~2023}
{Wang \MakeLowercase{\textit{et al.}}: Robust Beamforming Design and Antenna Selection for Dynamic HRIS-aided MISO Systems}


\maketitle

\begin{abstract}

In this paper, we propose a dynamic hybrid active-passive reconfigurable intelligent surface (HRIS) to enhance multiple-input-single-output (MISO) communications, leveraging the property of dynamically placing active elements. Specifically, considering the impact of hardware impairments (HWIs), we investigate channel-aware configurations of the receive antennas at the base station (BS) and the active/passive elements at the HRIS to improve transmission reliability. To this end, we address the average mean-square-error (MSE) minimization problem for the HRIS-aided MISO system by jointly optimizing the BS receive antenna selection matrix, the reflection phase coefficients, the reflection amplitude matrix, and the mode selection matrix of the HRIS. To overcome the non-convexity and intractability of this problem, we first transform the binary and discrete variables into continuous ones, and then propose a penalty-based exact block coordinate descent (PEBCD) algorithm to alternately solve these subproblems. Numerical simulations demonstrate the significant superiority of our proposed scheme over conventional benchmark schemes.
\end{abstract}


\section{Introduction}
\lettrine[lines=2]{M}{ASSIVE} multiple-input multiple-output (MIMO) has been a revolutionary solution for enhancing spectral efficiency (SE) to meet the rapidly growing traffic demand for future wireless communication systems. However, a fully digital implementation of MIMO requires costly dedicated radio-frequency (RF) chains for each antenna, inevitably resulting in prohibitive hardware costs and power consumption.

To improve both cost and power efficiency, analog signal processing can be additionally introduced to reduce the number of RF chains \cite{zhai2017joint}. One common technique is the so-called hybrid analog-digital transceiver, widely used in recent massive MIMO systems. The signal undergoes initial processing by a low-dimensional digital beamformer before traversing through a high-dimensional analog beamformer, constructed using analog phase shifters. Generally, there are two typical structures for implementing analog signal processing. By connecting the RF chains to all antennas or a subset of the antenna array, fully-connected or sub-connected hybrid analog-digital transceivers are utilized \cite{molisch2017hybrid}. However, hybrid beamforming still requires a large number of phase shifters, which are generally associated with complicated circuits and high energy consumption in practice. Another low-cost alternative technique is antenna selection \cite{Gao2018Hybrid}, where only a subarray of antennas is activated with the limited RF chains via the RF switch network. Compared to phase shifters, analog RF switches have much lower complexity and power consumption, making antenna selection more appealing \cite{elbir2019joint}.

{\color{blue} On a parallel track, advances in metasurfaces have promoted a new paradigm in wireless communications \cite{ataloglou2023metasurfaces}. For example, authors in \cite{StackedRIS} propose stacked intelligent metasurfaces (SIM) at the transmitter and receiver to implement holographic multiple-input multiple-output (HMIMO) communications without requiring excessive RF chains. Although metasurface-aided transceivers have huge potential for high energy efficiency, they still face challenges in the absence of line-of-sight scenarios. On the other hand, reconfigurable intelligent surface (RIS) \cite{wu2021intelligent} has exhibited remarkable capability to dynamically reshape the wireless propagation environment and explore new physical dimensions of transmission.}
By deploying a massive number of low-cost passive elements, RIS can dynamically adjust the phase of the incident signal to create favorable wireless channels, thereby improving communication performance. However, due to the fully passive nature of RIS elements, the widespread adoption of RIS has been greatly constrained. For example, the performance improvement induced by RIS is limited due to the "multiplicative fading" effect, particularly in the presence of a direct link. 

To overcome this challenge, the hybrid active-passive RIS (HRIS) has been proposed to compensate for the severe cascaded path loss with additional power amplifiers (PAs) \cite{peng2023hybrid,sankar2023channel}. 
{\color{blue} Unlike the amplify-and-forward (AF) relay, HRIS operates similarly to conventional RIS by directly reflecting the incident signal with the desired adjustments at the electromagnetic (EM) level. In contrast, an AF relay typically requires larger and more power-hungry RF chains to receive and then transmit the signal with amplification. This process usually occurs at the baseband level and requires two time slots to complete the amplify-and-forward cycle. Even when operating in full-duplex mode, the AF relay increases hardware complexity due to the need to mitigate self-interference. Additionally, the AF relay simply amplifies the received signal without modifying the phase.}
In \cite{peng2023hybrid}, the authors investigated the optimal ratio between the number of active and passive elements to maximize the ergodic capacity under the total power budget. On the other hand, the authors in \cite{sankar2023channel} have studied the channel-aware placement of the active elements to enhance the signal-to-noise ratio (SNR). Unfortunately, these existing works focused on the ideal hardware implementation and ignored the impact of hardware impairments (HWIs).

In this paper, we investigate a dynamic HRIS-assisted MISO communication system under the non-ideal hardware implementation
The channel-aware placement of the base station (BS) receive antennas and the active/passive HRIS elements are optimized to enhance the system performance. 
To this end, we investigate the average mean-square-error (MSE) minimization problem under the power budget of HRIS by jointly optimizing the receive antenna selection matrix, the reflection phase coefficients, the reflection amplitude matrix, and the mode selection matrix for the active and passive elements.
To tackle the mix-integer optimization problem, a penalty-based exact block coordinate descent (PEBCD) algorithm is proposed. Numerical simulations show great superiority in the channel-aware configuration of the BS receive antennas and the active/passive HRIS elements as compared to the conventional benchmarks.

    \begin{figure}[t]
    \centering  
     \includegraphics[width=0.4\textwidth]{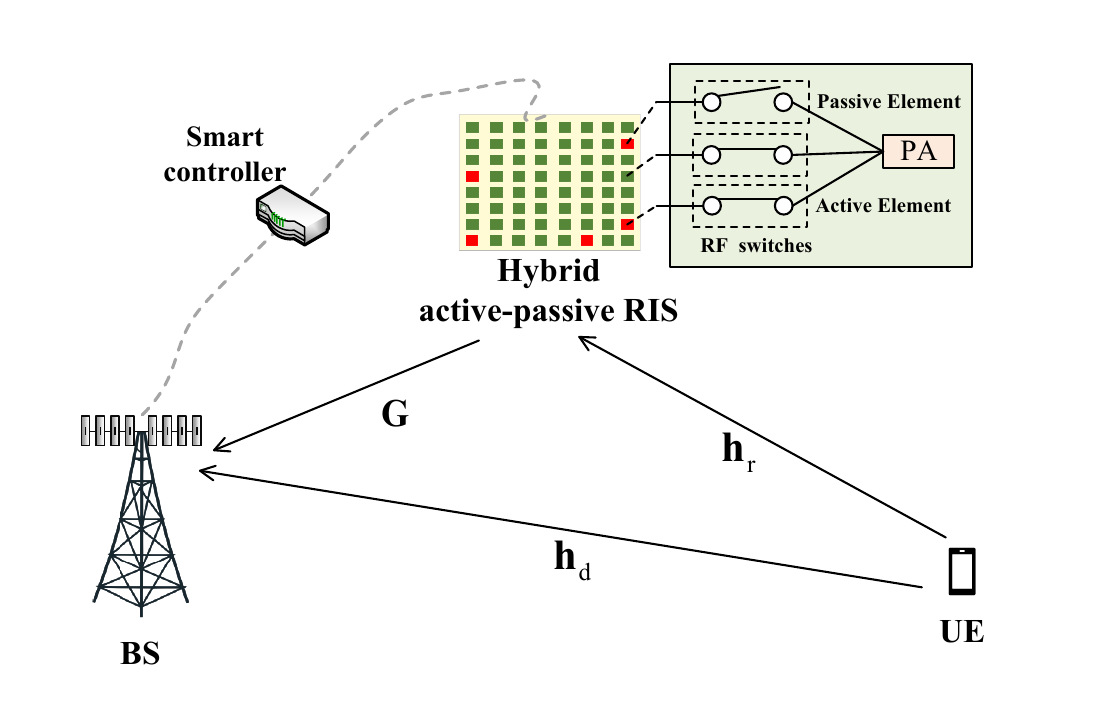}
     \vspace{-12pt}
        \caption{{\color{blue}A dynamic HRIS-assisted uplink MISO communication system.}}
        \vspace{-12pt}
        \label{system model}    
    \end{figure}

\section{System Model and Problem Formulation}
   
\subsection{System Model}
{\color{blue} As depicted in Fig.~\ref{system model}, a dynamic HRIS-aided uplink MISO communication system with a single-antenna user and an $N_R$-antenna BS is considered.
The HRIS comprises $N_a$ active and $N_p$ passive elements, where $N_a+N_p=N$. 
To reduce the power consumption and hardware complexity, we assume the BS receives the data stream from $N_R$ antennas with $L<N_R$ dedicated RF chains.
The antenna selection technique is employed to select $L$ out of $N_R$ antennas to explore the spatial diversity.}

We denote ${\bf{h}}_{d} \in \mathbb{C}^{N_R \times 1}$ and ${\bf{G}}\in \mathbb{C}^{N \times N_R}$ as the full channel state information (CSI) associated with $N_R$ receive antennas from the BS to the user and the HRIS, respectively. 
The channel between the HRIS and the user is ${\bf{h}}_{r} \!\! \in\!\! \mathbb{C}^{N \times 1}$.
Based on the receive antenna selection, $L$ out of $N_R$ or $N$ rows are selected to form the ${L \times 1}$ or ${L \times N}$ sub-channels, respectively.  
Denoting ${\bf{A}} \in \mathcal{A}$ as the index matrix for the $L$ selected antennas, where $\mathcal{A}\triangleq \{ {\bf{A}} \in \mathbb{C}^{L \times N_R}; A_{i,j} \in \{0,1 \}; \sum\nolimits_{j=1}^{N_R}A_{i,j}\!=\!1, \forall i; \sum\nolimits_{i=1}^{L}A_{i,j}\!\leq\! 1, \forall j \}$,  the sub-channels can be represented by ${\bf{A}}{\bf{h}}_{d} $ and ${\bf{A}}{\bf{G}} $.
All channels are assumed to experience quasi-static flat fading and perfectly available using the recent advances in RIS and MISO channel estimation.\footnote{\color{blue}{To delve into the unique effects of hardware impairments on the dynamic HRIS-assisted MISO communication system, we adopt the ideal CSI assumption, which is prevalent in contemporary studies \cite{CE1,shen2020beamforming}. The development of a channel estimation method falls outside the scope of this work and is therefore reserved for future exploration. It is crucial to highlight that the optimization framework we propose is designed to be easily adaptable to scenarios characterized by imperfect CSI. }}

{\color{blue} The dynamic HRIS, as shown in Fig.~1, is equipped with low-cost discrete phase shifters and RF switches. 
 To avoid the bulky circuits in practical implementation, we adopt the ``sub-connected architecture" for those active HRIS elements. This approach differs from the fully-connected one by employing a single power amplifier to serve all elements. The RF switches are used in conjunction with different phase-shift circuits, allowing HRIS elements to operate in either active or passive modes and independently control phase shifts while sharing a common amplification factor. }
The location and the number of active elements are optimized to enhance the system performance.
Specifically, we denote ${\bm{\Phi}} \triangleq {\rm diag}(e^{j \phi_1}, \dots,  e^{j \phi_N})$ as the reflection phase matrix, where $e^{j \phi_n}$ denotes the discrete phase shift of the $n$-th element with $\phi_n \in \mathcal{X}\triangleq\{0, \frac{2\pi}{2^B},\dots, \frac{2\pi(2^B-1)}{2^B}\}$. $B$ is the number of quantization bits. The reflection amplitude matrix of HRIS is denoted as ${\bf{\Lambda}}\triangleq {\rm diag}(\omega_1,\dots,\omega_N)$, where $\omega_n$ represents the amplification factor of $n$-th element. For the active elements, we have $\omega_n > {\mu}_{\rm min}$ when $n \in \mathcal{N}_{a}$, where ${\mu}_{\rm min}\geq 1$ denotes the minimum amplification factor and $\mathcal{N}_{a}$ is the index set of all active elements. 
Otherwise, we have $\omega_n =1$ for those passive elements. {\color{blue}With the same PA, all the active elements share one common amplification factor $\mu$.
We introduce a binary diagonal selection matrix ${\bf{B}}$ to establish the relationship between the HRIS elements and the PA, i.e., ${\bf{B}}={\rm diag}({\bm{\gamma}}) \in \mathcal{B}\triangleq \{ {\bm{\gamma}} \in \mathbb{C}^{N \times 1}; \gamma_i=1, \forall i \in\mathcal{N}_{a}; \gamma_i=0, \forall i \notin\mathcal{N}_{a} \}$.} 


Due to the non-ideal phase shifters of HRIS, the phase error of the $n$-th element can be modeled as ${\bar \phi}_n \in \mathcal{U}[-\frac{\pi}{2^B},\frac{\pi}{2^B}]$, where $\mathcal{U}[ \cdot]$ denotes the uniform distribution \cite{Wang2023HWI}. 
Therefore, the actual phase shift matrix with phase noise is $ {\bm{\tilde \Phi}}={{\bm{ \Phi}}}{\bm{\bar \Phi}}$, with ${\bm{\bar \Phi}}\triangleq {\rm diag}(e^{j {\bar \phi}_1}, \dots,  e^{j {\bar \phi}_N})$ denoting the phase noise matrix.
%
To accurately model the realistic imperfect hardware, the additive HWIs are considered. 
Specifically, the aggregate residual HWIs for the transceiver can be modeled as the additive Gaussian distortion noise \cite{gong2023hardware}.
As such, the actual transmitted signal for the user is expressed as
\begin{equation}
    {{\tilde s}}=\sqrt{p}{{s}}+{{\kappa}}_{t},
\end{equation}
where $s$ represents the transmitted symbol with $\mathbb{E}[s s^H]=1$. ${\bf{\kappa}}_{t} \!\sim\! \mathcal{CN}({{0}},{k_t^2}p)$ denotes the additive distortion noise whose variance is proportional to the transmit power $p$, where $k_t$ characterizes the distortion level. 

Then, the uplink received signal at the BS is written as
\vspace{-6pt}
\begin{align}
    {\bf{\tilde y}} = \underbrace{ {\bf{A}}({\bf{h}}_{d}\!+\!{\bf{G}}^H {\bf{\Lambda}} {\bm{\tilde \Phi}} {\bf{h}}_{r}){{\tilde s}} + {\bf{A}}{\bf{G}}^H {\bf{B}} {\bf{\Lambda}}{\bm{\tilde \Phi}} {\bf{n}}_a  + {\bf{n}}_{b} }_{{\bf{y}}} + {\bm{\kappa}}_{r},
\end{align}
where ${{\bf{y}}}$ denotes the un-distorted received signal and ${\bf{n}}_{b}$ stands for the additive white Gaussian noise (AWGN) whose elements are random variables with zero-mean and unit-variance, i.e., ${\bf{n}}_{b} \!\sim\! \mathcal{CN}({\bf{0}},{\sigma_b^2}{\bf{I}}_{L})$. 
 ${\bf{G}}^H {\bf{B}} {\bf{\Lambda}} {\bm{\tilde \Phi}}{\bf{n}}_a$ in ${{\bf{y}}}$ represents the amplified noise associated with the active elements, where ${\bf{n}}_a \!\sim\! \mathcal{CN}({\bf{0}},{\sigma_a^2}{\bf{I}}_N)$.
${\bm{\kappa}}_{r}$ stands for the receive distortion noise. We have ${\bm{\kappa}}_{r} \!\sim\! \mathcal{CN}(({\bf{0}},{k_r^2} {\rm \widetilde{diag}}(\mathbb{E}[{\bf{ y}}{\bf{ y}}^H ])) $, whose variance is proportional to the average received power of the individual antennas at the BS. 
{\color{blue}${\rm diag}(\cdot)$ returns a square diagonal matrix whose diagonal elements are the same as the input.}
$k_r$ represents the received distortion level and $\mathbb{E}[{\bf{ y}}{\bf{ y}}^H ]$ is derived as
\vspace{-6pt}
\begin{subequations}
\begin{align}
    \mathbb{E}[{\bf{y}}{\bf{y}}^H ]=& {\tilde p}{\bf{A}}\mathbb{E}_{{\bm{\bar \Phi}}}[{({\bf{h}}_{d}\!+\!{\bf{G}}^H {\bf{\Lambda}} {\bm{\tilde \Phi}} {\bf{h}}_{r}}){({\bf{h}}_{d}\!+\!{\bf{G}}^H {\bf{\Lambda}} {\bm{\tilde \Phi}} {\bf{h}}_{r})}^H]{\bf{A}}^H \nonumber \\
    & + \sigma_a^2{\bf{A}}  {\bf{G}}^H {\bf{B}} {\bf{\Lambda}}{\bm{\Phi}} {\bm{\Phi}}^H {\bf{\Lambda}}^H {\bf{B}}^H {\bf{G}}{\bf{A}}^H \!+\! \sigma_b^2 {\bf{I}}_{L}, \\
 \overset{(a)}{=} & {\bf{A}} {\bf{\Omega}} {\bf{A}}^H + \sigma_b^2 {\bf{I}}_{L}, 
 \end{align}
 \end{subequations}
 where ${\bf{\Omega}}={\tilde p} \left( {\bf{h}}{\bf{h}}^H \!+\! (1\!-\!\epsilon_b^2)  {\bf{G}}^H {\bf{\Lambda}} {\bm{\Phi}} {\rm \widetilde{diag}}({{\bf{h}}_{r}}{{\bf{h}}_{r}}^H) ({\bm{\Phi}} {\bf{\Lambda}})^H {\bf{G}}
 \right) + \sigma_a^2  {\bf{G}}^H {\bf{B}} {\bf{\Lambda}} {\bf{\Lambda}}^H {\bf{B}}^H {\bf{G}}$ with ${\tilde p}=p(1+k_t^2)$ and $\epsilon_b=\frac{\sin({\pi/2^B})}{ {\pi/2^B} }$. 
 ${\bf{h}}={\bf{h}}_{d}+\epsilon_b {\bf{G}}^H {\bf{\Lambda}} {\bm{\Phi}} {\bf{h}}_{r}$ denotes the average effective channel between the user and the BS.
 The equality $(a)$ holds since $\mathbb{E}[{\bm{\bar \Phi}} {\bf{A}} {\bm{\bar \Phi}}^H] = \epsilon_b^2 {\bf{A}} + (1\!-\!\epsilon_b^2) {\rm \widetilde{diag}}({\bf{A}})$ and $\mathbb{E}[{\bm{\bar \Phi}} {\bf{A}} ] = \epsilon_b {\bf{A}}$.

Assume that the receive beamforming at the BS is ${\bf{w}}$, then the estimated symbol can be denoted as 
\begin{align}
    {\hat s}=& \sqrt{p}{\bf{w}}^H  {\bf{A}}  {\bf{\tilde h}} {s} \!+\! {\bf{w}}^H  {\bf{A}}  {\bf{\tilde h}}{{\kappa}}_{t} \!+\! {\bf{w}}^H  {\bf{A}} {\bf{G}}^H {\bf{B}} {\bf{\Lambda}}{\bm{\tilde \Phi}}{\bf{n}}_a \nonumber \\
    &\!+\! {\bf{w}}^H{\bf{n}}_b  \!+\! {\bf{w}}^H{\bm{\kappa}}_{r},
\end{align}
with ${\bf{\tilde h}} ={\bf{h}}_{d}+{\bf{G}}^H {\bf{\Lambda}} {\bm{\tilde \Phi}} {\bf{h}}_{r}$.
Therefore, the average MSE\footnote{ {\color{blue}Due to the correlation between the minimum MSE and the minimum average bit error rate (BER), the MSE metric provide a elegant expression for measuring the quality of service compared to the BER, which is difficult or intractable due to their dependence on the integral $Q$ function \cite{tenenbaum2011mse}.}
} can be derived as
\begin{align}
    f_{\rm MSE} & = \mathbb{E}[({\hat s}-{s})^2]  = {\bf{w}}^H {\bf{Q}}  {\bf{w}}  - 2\sqrt{p}\Re\{ {\bf{w}}^H{\bf{A}} {\bf{h}} \} + 1,
\end{align}
with ${\bf{Q}}\!=\!{\bf{A}}{\bf{\Omega}} {\bf{A}}^H \!+\! k_r^2  {\rm \widetilde{diag}} ({\bf{A}}{\bf{\Omega}} {\bf{A}}^H )\!+\! {\tilde \sigma}_b^2  {\bf{I}}_{L}$ and ${\tilde \sigma}_b^2\!=\!\sigma_b^2(1+k_r^2)$.

\subsection{Problem formulation}
In this paper, we mainly focus on enhancing transmission reliability of the considered system in terms of MSE. Specifically, the average MSE is minimized under the total power budget $P_{\rm HRIS}$ of HRIS by jointly optimizing the receive antenna selection matrix ${\bf{A}}$, the mode selection matrix ${\bf{B}}$, the reflection phase matrix ${\bf{\Phi}}$, and the reflection amplitude matrix ${\bf{\Lambda}}$.  
The total power consumption at the HRIS is calculated as
    $P  \!=\!  \mathbb{E}[|{\bf{\Lambda}}{\bf{B}} {\bm{\tilde \Phi}} {\bf{h}}_{r} {\tilde s} |^2] \!+\!  \mathbb{E} [|  {\bf{B}} {\bf{\Lambda}} {\bm{\tilde \Phi}} {\bf{n}}_{a} |^2 ] =   {\tilde p} |{\bf{\Lambda}}{\bf{B}}  {\bf{h}}_{r}  |^2 + \sigma_a^2 |  {\bf{B}} {\bf{\Lambda}}  |^2$.
Then, the average MSE minimization problem can be formulated as
\begin{subequations} \label{Obj_1}
\begin{align} 
    {\textbf{(P1)}}:
    &{\min_{ \substack{ {\bf{w}},{\bf{A},{\bf{\Phi}}},{\bf{\Lambda}}, {\bf{B}} } }} \quad   f_{\rm MSE}  \\
    
     &~~~~~{\mbox{s.t.}} \quad \quad~ 
     {\tilde p} |{\bf{\Lambda}}{\bf{B}}  {\bf{h}}_{r}  |^2 + \sigma_a^2 |  {\bf{B}} {\bf{\Lambda}}  |^2 \leq P_{\rm HRIS}, \label{P1_power}\\
    
    &\quad \qquad \qquad \phi_n \in \mathcal{X},
    \label{P1_theta}\\
    
    & \quad \qquad \qquad \omega_n > \mu_{\rm min}, \forall n \in \mathcal{N}_{a}, \label{P1_ampfac} \\
    
    
    
    & \quad \qquad \qquad {\bf{A}} \in \mathcal{A}, {\bf{B}} \in \mathcal{B},  \label{P1_A} 
    
\end{align}
\end{subequations}
where \eqref{P1_theta} denotes the discrete phase coefficient constraint and \eqref{P1_ampfac} limits the amplification factor of the active elements. 
Unfortunately, the problem (P1) is intractable and NP-hard due to the integer constraints and tightly coupled variables. 
In the sequel, we first transform these intractable constraints and then present a PEBCD algorithm to obtain a stationary solution.

\vspace{-6pt} 
\section{Proposed PEBCD algorithm}
In this section, we begin by transforming the binary and discrete variables into continuous ones and then propose to optimize these variables with the others fixed alternately.
\vspace{-6pt}
\subsection{Problem Transformation}
Since the reflection amplitude matrix ${\bf{\Lambda}}\!=\!{\rm diag}({\bm{\omega}})$ and the mode selection matrix ${\bf{B}}\!=\!{\rm diag}({\bm{\gamma}})$ are tightly coupled, {\color{blue}we first introduce the auxiliary variable $ {\bm{\tilde \omega}}={\rm diag}({\bf{\Lambda}}{\bf{B}})$, representing the reflection amplification matrix related to those active HRIS elements, as follows}
\begin{align}
{{\tilde \omega}}_i = \left\{\begin{array}{l}
{\mu} \geq \mu_{\rm min}, \quad  i \in \mathcal{N}_{a}, \\ 0, \quad i \notin \mathcal{N}_{a}, 
\end{array}\right.
\end{align}
where $\mu$ denotes the common amplification factor for the active elements.
Thus, we have ${\bm{\tilde \omega}}=\mu{\bm{\gamma}}$ and ${\bm{ \omega}}=(\mu-1){\bm{\gamma}} + {\bm{1}}$.
The power constraint in \eqref{P1_power} can be transformed into 
   $\mu^2 {\bm{\gamma}}^H ({\tilde p} { \rm \widetilde{diag}}({\bf{h}}_r{\bf{h}}_r^H) +\sigma_a^2 {\bf{I}}) {\bm{\gamma}} \leq P_{\rm HRIS}$.
Therefore, we equivalently optimize the common amplification factor $\mu$ and the mode selection vector ${\bm{\gamma}}$ in the following.

For the receive antenna selection matrix ${\bf{A}}$, we define ${\bf{a}} \triangleq {\rm vec}({\bf{A}}^T)$ with ${\bf{a}}=\left[ {\bf{a}}_1^T, {\bf{a}}_2^T, \dots, {\bf{a}}_L^T \right]^T$ and ${\bf{a}}_i \neq {\bf{a}}_j, \forall i \neq j$.
${\bf{a}}_i \in \mathcal{S}_{a}\triangleq \{ {\bf{a}}\in \mathbb{R}^{N_r} | {\bf{1}}^T{\bf{a}}=1, [{\bf{a}}]_m \in \{0,1\},\forall m \}$ denotes each ${\bf{a}}_i$ belongs to a special ordered set of type 1 (SOS1).   
On the other hand, to efficiently deal with the discrete phase coefficients, we adopt the binary representation for the discrete variable ${\bm{\theta}}\triangleq {\rm diag}({\bm{\Phi}})$. 
Defining 
${\bm{\theta}}_{s}$ as ${\bm{\theta}}_{s}=[1,e^{j\frac{2\pi}{2^B}},\dots, \frac{2\pi(2^B-1)}{2^B}]^T$, we can rewrite ${\bm{\theta}}$ as ${\bm{\theta}}={\bf{Z}}{\bm{\theta}}_{s}$, where ${\bf{Z}}=[{\bf{z}}_1,{\bf{z}}_2,\dots,{\bf{z}}_N]\in\mathbb{C}^{N\times2^B}$. Each ${\bf{z}}_n$ belongs to the SOS1, i.e., ${\bf{z}}_n\in \mathcal{S}_{z}\triangleq \{ {\bf{z}}\in \mathbb{R}^{2^B} | {\bf{1}}^T{\bf{z}}=1, [{\bf{z}}]_m \in \{0,1\},\forall m \}$. Similar with the receive antenna selection matrix ${\bf{A}}$, we further define 
${\bf{z}} \triangleq {\rm vec}({\bf{Z}}^T)$ with ${\bf{z}}=\left[ {\bf{z}}_1^T, {\bf{z}}_2^T, \dots, {\bf{z}}_L^T \right]^T$.

The binary variables $\{ {{\bm{\gamma}},{\bf{a}}, {\bf{z}}} \}$ still render the problem intractable and NP-hard. 
Though the optimal solution can be obtained by the state-of-the-art Gurobi solver embedded with the branch-and-bound (BB) method, the worst-case complexity increases exponentially with the scale of the problem. 
To balance the complexity and optimality, we propose a variational reformulation of the binary constraints. Specifically, these constraints can be equivalently transformed into the $l_2$ box non-separable constraints with the following lemma.
\begin{lemma} \label{binary}
Assume ${\bf{x}} \in \mathbb{R}^N,{\bf{y}} \in \mathbb{R}^N$ and define $\chi \triangleq \{ ({\bf{x}},{\bf{y}}) | {\bf{-1}} \leq {\bf{x}} \leq {\bf{1}}, ||{\bf{y}}||_2^2 \leq N, {\bf{x}}^T{\bf{y}}=N, \forall {\bf{y}} \}$. Assume that ${\bf{x}},{\bf{y}} \in \chi$, then we have ${\bf{x}} \in \{-1,1\}^N$, and ${\bf{x}}={\bf{y}}$.
\end{lemma}
\begin{proof}
See \cite{yuan2017exact} for detailed proof.
\end{proof}

Based on Lemma~\ref{binary} and introducing the auxiliary variables $\{{\bf{\tilde u}},{\bf{\tilde v}}, {\bf{\tilde q}}\}$, the problem (P1) equivalently turns into the following mix-integer optimization problem:

\begin{subequations} \label{Obj_3}
\begin{align} 
    {\textbf{(P2)}}:\!\!\! \min_{ \substack{ {\bf{w}},{\mu},{\bf{v}}, {\bf{u}}, {\bf{q}} \\{\bf{0}} \leq \{ {\bm{\gamma}},{\bf{a}},{\bf{z}} \} \leq {\bf{1}}   }} \!\!\!   &f_{\rm MSE}  \\
    
     \qquad \mbox{s.t.} \quad   
     &\mu^2 {\bm{\gamma}}^H ({\tilde p} { \rm \widetilde{diag}}({\bf{h}}_r{\bf{h}}_r^H) \!+\! \sigma_a^2 {\bf{I}}) {\bm{\gamma}} \leq P_{\rm HRIS}, \label{Power_HRIS} \\ 
     
      \quad \qquad  \quad &{\bf{a}}_i \neq {\bf{a}}_j, \forall i \neq j,\{i,j\}\in \mathcal{L}, \label{a_01}  \\
     
      \quad \qquad  \quad &\mu >\mu_{\rm min}, \label{mu_HRIS} \\
     
    \quad \qquad  \quad &{\bm{\tilde \gamma}}^T  {\bf{\tilde u}} \!=\! N,{\bf{\tilde a}}^T  {\bf{\tilde v}} \!=\! LN_r, {\bf{\tilde z}}^T  {\bf{\tilde q}} \!=\! 2^BN, \label{P3_E1} \\

    \quad \qquad  \quad &||{\bf{\tilde u}}||_2^2 \!\leq\!  N,||{\bf{\tilde v}}||_2^2 \!\leq\!  LN_r, ||{\bf{\tilde q}}||_2^2 \!\leq\!  2^BN, \label{P3_InE1} \\

    \quad \qquad  \quad &{\bf{1}}^T{\bf{z}}_n = 1,\forall n\in \mathcal{N}; {\bf{1}}^T{\bf{a}}_i = 1,\forall i\in \mathcal{L}, \label{P3_E2}
\end{align}
\end{subequations}
where ${\bm{\tilde \gamma}}\triangleq2{\bm{\gamma}}-{\bf{1}}$ and $\{ {\bf{\tilde a}},{\bf{\tilde z}},{\bf{\tilde v}}, {\bf{\tilde u}}, {\bf{\tilde q}}\}$ are similarly defined w.r.t. $\{ {\bf{a}},{\bf{z}},{\bf{v}}, {\bf{u}}, {\bf{q}}\}$.

To efficiently tackle the equality constraints in \eqref{P3_E1}, we penalize the complementary error directly by a penalty function. The resultant objective function is defined as
\begin{align}
    \mathcal{L}_{\rho} & =  f_{\rm MSE} + \mathcal{J}_{\rho},
\end{align}
where the penalty term $\mathcal{J}_{\rho}={\rho}[ (N\!-\!{\bm{\tilde \gamma}}^T  {\bf{\tilde u}}) \!+\! ( LN_r\!-\!{\bf{\tilde a}}^T  {\bf{\tilde v}}) \!+\! 
    ( 2^BN \!-\!{\bf{\tilde z}}^T  {\bf{\tilde q}}) ]$.
Thus, the problem (P3) turns into
\begin{align}  \label{Obj_4}
    {\textbf{(P3)}}: &\min_{ \substack{ {\bf{w}},{\mu},{\bf{v}}, {\bf{u}}, {\bf{q}} \\{\bf{0}} \leq \{ {\bm{\gamma}},{\bf{a}},{\bf{z}} \} \leq {\bf{1}}   }}     \mathcal{L}_{\rho} \quad  \mbox{s.t.} ~   
     \eqref{Power_HRIS},\eqref{a_01},\eqref{mu_HRIS}, \eqref{P3_InE1},\eqref{P3_E2}.
\end{align}
In the sequel, we alternatively optimize these variables with the remaining ones fixed.

\vspace{-8pt}
\subsection{Receive Beamforming}
For any given $\{{\mu},{\bf{v}}, {\bf{u}}, {\bf{q}},  {\bm{\gamma}},{\bf{a}},{\bf{z}}  \}$, the linear MMSE detector is the optimal receive beamforming. 
Thus, the optimal receive beamforming ${\bf{w}}^{\star}$ is written as 
\begin{align}\label{receive beamforming}
    {\bf{w}}^{\star} = \sqrt{p} \left({\bf{A}}{\bf{\Omega}} {\bf{A}}^H \!+\! k_r^2  {\rm \widetilde{diag}} ({\bf{A}}{\bf{\Omega}} {\bf{A}}^H )\!+\! {\tilde \sigma}_b^2  {\bf{I}}_{L} \right) ^{-1} {\bf{A}} {\bf{h}}.  
\end{align}

\vspace{-8pt}
\subsection{Amplification Factor }
In this subsection, we optimize the common amplification factor for the active elements of the HRIS with the other variables fixed. 
Since the penalty term $\mathcal{J}_{\rho}$ is irrelevant to the amplification factor $\mu$, we first rewrite the objective $f_{\rm MSE}$  as
     $f_{\rm MSE}  \!=\! f_{{\bm{\omega}},{\bm{\tilde \omega}}} \!+\! {\tilde p} {\bf{h}}_{d}^H {\bf{X}} {\bf{h}}_{d} \!+\! {\tilde \sigma}_b^2{\bf{w}}^H{\bf{w}} \!-\! 2\sqrt{p} \Re \{ {\bf{h}}_{d}^H {\bf{A}}^H {\bf{w}} \} \!+\! 1$
utilizing the matrix identity of  ${\rm Tr}({\rm {diag}}({\bf{a}})^H {\bf{B}} ) \!=\!{\bf{a}}^H {\rm {diag}}  ({\bf{B}}) $ and ${\rm Tr}({\bf{A}}{\rm {diag}}({\bf{b}}){\bf{C}}{\rm {diag}}({\bf{b}})^H )\!=\!{\bf{b}}^H ({\bf{C}}^T \odot {\bf{A}}){\bf{b}}$.
And $f_{{\bm{\omega}},{\bm{\tilde \omega}}}$ is defined as 
\begin{align}
    f_{{\bm{\omega}},{\bm{\tilde \omega}}} \! =\!   {\bm{\tilde \omega}}^H {\rm \widetilde{diag}}( \sigma_a^2 {\bf{G}}{\bf{X}}{\bf{G}}^H ){\bm{\tilde \omega}} \!+\! {\bm{\omega}}^H {\bf{K}} {\bm{\omega}} \!+\! 2\Re \{ {\bm{\omega}}^H {\bf{k}}\},
\end{align}
with the auxiliary terms ${\bf{k}}={\tilde p}\epsilon_b {\rm diag}({\bf{\Phi}}^H {\bf{G}} {\bf{X}} {\bf{h}}_{d}{\bf{h}}_{r}^H) - \sqrt{p} \epsilon_b  {\rm diag}({\bf{\Phi}}^H {\bf{G}} {\bf{A}}^H {\bf{w}}{\bf{h}}_{r}^H)$ and ${\bf{K}}=({\bf{\Phi}}^H {\bf{G}}{\bf{X}}{\bf{G}}^H {\bf{\Phi}}) \odot ({\tilde p}\epsilon_b^2 {{\bf{h}}_{r}} {{\bf{h}}_{r}}^H + {\tilde p}(1-\epsilon_b^2){\rm \widetilde{diag}}({{\bf{h}}_{r}}{{\bf{h}}_{r}}^H))^T$.
${\bf{X}}$ is defined as
${\bf{X}}  ={\bf{A}}^H ( {\bf{w}} {\bf{w}}^H + k_r^2 {\rm \widetilde{diag}} ({\bf{w}} {\bf{w}}^H) ) {\bf{A}}$.

Substituting ${\bm{\tilde \omega}}=\mu{\bm{\gamma}}$ and ${\bm{ \omega}}=(\mu\!-\!1){\bm{\gamma}} \!+\! {\bm{1}}$ into $f_{{\bm{\omega}},{\bm{\tilde \omega}}}$, it can be further transformed to $f_{{\bm{\gamma}},{\mu}}$ as follows:
\begin{align} \label{f_gamma_mu}
    f_{{\bm{\gamma}},{\mu}}  =& \sigma_a^2{\mu}^2{\bm{\gamma}}^H {\rm \widetilde{diag}}(  {\bf{G}}{\bf{X}}{\bf{G}}^H ){\bm{\gamma}}\! + \!({\mu}\!-\!1)^2{\bm{\gamma}}^H  {\bf{K}} {\bm{\gamma}} \nonumber \\
    & ~ \!+\! 2({\mu}\!-\!1) {\bm{\gamma}}^H \Re \{ {\bf{K}}{\bf{1}}\!+\!{\bf{k}}\} \!+\! {\bf{1}}^T{\bf{K}}{\bf{1}} \!+\!  2\Re \{ {\bf{k}}^H{\bf{1}}\}.
\end{align}
With the fixed ${\bm{\gamma}}$, the problem for optimizing $\mu$ at $(t\!+\!1)$-th iteration can be formulated as
\begin{align} 
     {\textbf{(P4)}}: \min_{{\mu}>{\mu}_{\rm min}} \quad &  a \mu^2 + 2b \mu  \quad  
    \mbox{s.t.} \quad  \mu \leq \mu_{\rm ref},
\end{align}
where the parameters $a,b$ and $\mu_{\rm ref}$ are computed as
\begin{subequations}
\begin{align}
    a & =   \sigma_a^2 {\bm{\gamma}}_{t}^H {\rm \widetilde{diag}}({\bf{G}}{\bf{X}}{\bf{G}}^H) {\bm{\gamma}}_{t} + {\bm{\gamma}}_{t}^H{\bf{K}}{\bm{\gamma}}_{t}, \\
    b & =  {\bm{\gamma}}_{t}^H \Re \{ {\bf{K}}{\bf{1}}+{\bf{k}}\} - {\bm{\gamma}}_{t}^H{\bf{K}}{\bm{\gamma}}_{t},  \\
    \mu_{\rm ref} & = \sqrt{\frac{ P_{\rm HRIS} } { {\bm{\gamma}}_{t}^H ({\tilde p} { \rm \widetilde{diag}}({\bf{h}}_r{\bf{h}}_r^H) +\sigma_a^2 {\bf{I}}) {\bm{\gamma}}_{t} }}. 
\end{align}  
\end{subequations}
{\color{blue}The $\mu_{\rm ref}$ represents the maximum amplification factor that is constrained by the transmit power allocated to the active elements at the current iteration.}
Then we have the following proposition concerning the amplification factor selection:
\begin{proposition}
The HRIS should employ passive elements only when the power budget $P_{\rm HRIS}$ satisfies
\begin{align}
     P_{\rm HRIS} < P_{\rm min}=\mu_{\rm min}^2  ({\tilde p} |{{h}}_{r}^{\rm min}|^2 +\sigma_a^2 ),
\end{align}
with $|{{h}}_{r}^{\rm min}|\!=\!\min \{ |{h}_{r,n}|\}$ and employ active elements otherwise.
\end{proposition}
\begin{proof}
    The consumed power at the HRIS achieves the minimum when the amplification factor is equal to $\mu_{\rm min}$ and the number of active elements equals 1. Thus, the consumed power at the HRIS satisfies $P\!=\!\mu_{\rm min}^2  ({\tilde p} |{{h}}_{r,n}|^2 +\sigma_a^2 )\!\geq\! P_{\rm min}$.
\end{proof}
Then, we consider the non-trivial case with $ P_{\rm HRIS} \geq P_{\rm min} $, where the number of active elements $N_{\rm act}>0$. 
The optimal amplification factor for the unconstrained minimization in problem (P4) can be derived as
\begin{align}
    \tilde \mu = \frac{ {\bm{\gamma}}^H{\bf{K}}{\bm{\gamma}}-{\bm{\gamma}}^H \Re \{ {\bf{K}}{\bf{1}}+{\bf{k}}\} }{ \sigma_a^2 {\bm{\gamma}}^H {\rm \widetilde{diag}}({\bf{G}}{\bf{X}}{\bf{G}}^H) {\bm{\gamma}} + {\bm{\gamma}}^H{\bf{K}}{\bm{\gamma}}}.
\end{align}

Thus, the optimal $\mu^{\star}$ for problem (P4) can be obtained by comparing $\tilde \mu$ and $\mu_{\rm ref}$ as follows:
\begin{align}\label{mu opt}
    \mu^{\star} = \left\{\begin{array}{l}
\mu_{\rm min}, \quad  \tilde \mu \leq \mu_{\rm min}, 
\\ \tilde \mu , \quad   \mu_{\rm min}<\tilde \mu<\mu_{\rm ref}, 
\\  \mu_{\rm ref},  \quad   \tilde \mu \geq \mu_{\rm ref}.
\end{array}\right.
\end{align}

\begin{remark}
For the first two cases with $\tilde \mu \leq \mu_{\rm min}$ and $   \mu_{\rm min}<\tilde \mu<\mu_{\rm ref} $, the power constraint of the active elements is inactive, indicating the power budget at the HRIS is sufficient. The system performance will degrade if the amplification factor is greater than the threshold $\mu^{\star}$ due to the amplified noise and hardware impairments.
\end{remark}
\begin{remark}
The power constraint becomes active in the case of $\tilde \mu>\mu_{\rm ref}$. Thus, the system performance can be further enhanced with a larger power budget.
\end{remark}

\vspace{-8pt}
\subsection{Auxiliary Variables}
For the auxiliary variables ${ {\bf{u}}, {\bf{v}},{\bf{q}}}$, the problem turns into the following individual convex optimization problems:
\vspace{-3pt}
\begin{subequations}\label{aux var}
\begin{align}
    {\bf{u}} & := \arg \max_{ ||2{\bf{u}}-{\bf{1}}||_2^2\leq N }   (2{\bm{\gamma}}-{\bf{1}})^T (2{\bf{u}}-{\bf{1}}),   \label{L_u} \\
    
    {\bf{v}} & := \arg \max_{ ||2{\bf{v}}-{\bf{1}}||_2^2\leq LN_r } (2{\bf{a}}-{\bf{1}})^T (2{\bf{v}}-{\bf{1}}),  \label{L_v} \\
    
    {\bf{q}} & := \arg \max_{ ||2{\bf{q}}-{\bf{1}}||_2^2\leq 2^BN } (2{\bf{z}}-{\bf{1}})^T (2{\bf{q}}-{\bf{1}}). \label{L_gamma} 
\end{align}
\end{subequations}

The optimal closed-form solutions ${ {\bf{u}}^{\star}, {\bf{v}}^{\star},{\bf{q}}}^{\star}$ are readily derived as
\vspace{-6pt}
\begin{subequations}
\begin{align}
    {\bf{u}}^{\star} & ={\sqrt{N}} ({\bm{\gamma}}-0.5{\bf{1}})/{||2{\bm{\gamma}}-{\bf{1}} ||} + 0.5{\bf{1}}, \\
    {\bf{v}}^{\star} & ={\sqrt{LN_r}} ({\bf{a}}-0.5{\bf{1}})/{||2{\bf{a}}-{\bf{1}} ||} + 0.5{\bf{1}}, \\
    {\bf{q}}^{\star} & ={\sqrt{2^BN}} ({\bf{z}}-0.5{\bf{1}})/{||2{\bf{z}}-{\bf{1}} ||} + 0.5{\bf{1}}.
\end{align}
\end{subequations}

\vspace{-8pt}
\subsection{Selection Variables}
In this subsection, we alternatively optimize the dynamic mode selection vector ${\bm{\gamma}}$, the receive antenna selection vector ${\bf{a}}$ and the reflection phase selection vector ${\bf{z}}$ using the same optimization oracles. 

\subsubsection{Dynamic mode selection}

Based on the transformed objective function $f_{{\bm{\gamma}},{\mu}}$ in \eqref{f_gamma_mu}, the dynamic mode selection problem can be equivalently formulated as the following convex quadratic optimization problem with the quadratic and linear constraints:
\begin{align}\label{mode sel} 
    {\textbf{(P5)}}: \min_{ {\bm{\gamma}} } \quad &  {\bm{\gamma}}^T {\bf{E}}_{1} {\bm{\gamma}} + {\bm{\gamma}}^T {\bf{e}}\!-\!{\rho}(2{\bm{\gamma}}\!-\!{\bf{1}})^T(2{\bf{u}}-{\bf{1}}) \nonumber \\ 
    \mbox{s.t.} \quad &   {\bm{\gamma}}^T {\bf{E}}_{2} {\bm{\gamma}} \leq P_{\rm HRIS}, {\bf{0}}\leq {\bm{\gamma}} \leq {\bf{1}},
\end{align}
with ${\bf{E}}_{1}= {\mu}^2 \sigma_a^2 {\rm \widetilde{diag}}({\bf{G}}{\bf{X}}{\bf{G}}^H)  + ({\mu}-1)^2  {\bf{K}}$, ${\bf{E}}_{2}= \mu^2  ({\tilde p} { \rm \widetilde{diag}}({\bf{h}}_r{\bf{h}}_r^H) +\sigma_a^2 {\bf{I}})$ and ${\bf{e}}= \Re \{ 2({\mu}-1)({\bf{K}}{\bf{1}}+{\bf{k}})\}$.

\subsubsection{Receive antenna selection}
With the matrix identity of ${\rm Tr}({\bf{A}}{\bf{B}}{\bf{C}}{\bf{D}})=({\rm vec}({\bf{D}}^T))^T ({\bf{C}}^T \otimes {\bf{A}}){\rm vec}({\bf{B}})$ and ${\rm Tr}({\bf{A}}{\rm \widetilde{diag}}({\bf{B}}) ) = {\rm Tr}({\rm \widetilde{diag}}({\bf{A}}) {\bf{B}} )$, the objective $f_{\rm MSE}$ can be equivalently transformed as 
$f_{\rm MSE}=f_{{\bf{a}}}+ {\tilde \sigma}_b^2{\bf{w}}^H{\bf{w}} + 1$
where $ f_{{\bf{a}}} = {\bf{a}}^T {\bf{M}} {\bf{a}} - 2\Re\{ {\bf{a}}^T {\bf{m}}\} $. ${\bf{M}}$ and ${\bf{m}}$ are defined as $ {\bf{M}}=({\bf{w}}{\bf{w}}^H)^T \otimes {\bf{\Omega}}+ k_r^2 {\rm \widetilde{diag}} ({\bf{w}} {\bf{w}}^H ) \otimes {\bf{\Omega}} $ and ${\bf{m}}=\sqrt{p} {\rm vec}({\bf{h}}{\bf{w}}^H )$.
Then, the receive antenna selection  optimization problem can be formulated as the following convex optimization problem with the quadratic objective and linear constraints:
\begin{align} \label{ant sel}
   {\textbf{(P6)}}: \min_{{\bf{a}} } 
    
    \quad &  {\bf{a}}^T {\bf{M}} {\bf{a}} \!-\! 2\Re\{ {\bf{a}}^T {\bf{m}}\} \!-\!{\rho}(2{\bf{a}}\!-\!{\bf{1}})^T(2{\bf{v}}-{\bf{1}})  \nonumber \\
    
    \mbox{s.t.} 
    
    \quad & \eqref{a_01},{\bf{0}}\leq {\bf{a}} \leq {\bf{1}},{\bf{1}}^T{\bf{a}}_i = 1,\forall i\in \mathcal{L}.
\end{align}

\subsubsection{Reflection phase selection}
For the reflection phase optimization, we first transform the objective function into a tractable form w.r.t. ${\bm{\theta}}$ and then rewrite it using the binary representative variables ${\bf{z}}$.
We can decouple the variable ${\bm{\theta}}$ from the objective function $f_{\rm MSE}$ in \eqref{Obj_3} as 
        $f_{\rm MSE}  = f_{{\bm{\theta}}} \!+\!g$,
where $g$ denotes the constant term w.r.t. ${\bm{\theta}}$.
$f_{{\bm{\theta}}}$ is defined as
$f_{{\bm{\theta}}}  = {\bm{\theta}}^H {\bf{N}} {\bm{\theta}} + 2\Re \{ {\bm{\theta}}^H {\bf{n}}\}$
where  ${\bf{n}}={\tilde p}\epsilon_b {\rm diag}({\bf{\Lambda}}^H {\bf{G}} {\bf{X}} {\bf{h}}_{d}{\bf{h}}_{r}^H) - \sqrt{p} \epsilon_b  {\rm diag}({\bf{\Lambda}}^H {\bf{G}} {\bf{A}}^H {\bf{w}}{\bf{h}}_{r}^H)$ and   
$
{\bf{N}}  =({\bf{\Lambda}}^H {\bf{G}}{\bf{X}}{\bf{G}}^H {\bf{\Lambda}}) \odot ({\tilde p}\epsilon_b^2 {{\bf{h}}_{r}} {{\bf{h}}_{r}}^H + {\tilde p}(1-\epsilon_b^2){\rm \widetilde{diag}}({{\bf{h}}_{r}}{{\bf{h}}_{r}}^H)) 
$. 
Recalling ${\bm{\theta}}={\bf{Z}}{\bm{\theta}}_{s}$ and  
${\bf{z}} = {\rm vec}({\bf{Z}}^T)$, the reflection phase optimization problem can be equivalently transformed into the following convex optimization problem: 
\begin{align}\label{dis phase}
    {\textbf{(P7)}}: \min_{{\bf{z}}} \quad &  {\bf{z}}^T {\bf{\tilde N}} {\bf{z}} +  2\Re\{ {\bf{z}}^T{\bf{\tilde n}}\} \!-\!{\rho}(2{\bf{z}}\!-\!{\bf{1}})^T(2{\bf{q}}-{\bf{1}})\nonumber \\
    \mbox{s.t.} \quad & {\bf{0}}\leq {\bf{z}} \leq {\bf{1}},{\bf{1}}^T{\bf{z}}_n = 1,\forall n\in \mathcal{N},
\end{align}
where ${\bf{\tilde N}}={\bf{N}}^T \otimes {\bm{\theta}}_{s}{\bm{\theta}}_{s}^H$ and ${\bf{\tilde n}}={\rm vec}( {\bm{\theta}}_{s} {\bf{n}}^H)$.

Since the optimization problems (P5–7) are quadratic convex problems subject to linear and quadratic constraints, the global optimal solutions can be efficiently obtained via the standard CVX solvers.
The penalty $\rho$ is updated with every $T$ iterations to realize a trade-off between the solution accuracy and the computational complexity. 

{\color{blue}Finally, the PEBCD algorithm is summarized in {\bf{Algorithm 1}}. The proposed algorithm is guaranteed to converge to a stationary solution with a sufficiently large penalty. The computational complexity of Algorithm 1 primarily stems from the update of primal variables in step 3, 4, and 6. In step 3, the computation of matrix inversion and multiplication contribute to a complexity of $O(LN_R^2+L^2N_R+L^3)$. Step 4 involves the matrix multiplication with a complexity of $O(N^2N_R+NN_R^2)$. The interior point method of the CVX solver has a complex of $O(\log(1/{\epsilon})(2^{3B}N^3+L^3N_R^3))$, with $\epsilon$ representing the desired accuracy. Therefore, the total complexity of {\bf{Algorithm 1}} can be estimated as $O(I_{BCD}(\log(1/{\epsilon})(2^{3B}N^3+L^3N_R^3)+N^2N_R+NN_R^2))$, where $I_{BCD}$ denotes the number of iterations for the proposed PEBCD algorithm. }
\vspace{-6pt}
\begin{algorithm} \label{Proposed PEBCD Algorithm}
    \normalsize
      \caption{Proposed PEBCD Algorithm }
     \SetKwInOut{Input}{Input}
     \SetKwInOut{Output}{Output}
     \Input{System parameters  $N_R$, $L$, $N$, $B$, $k_t$, $k_r$, $\sigma_b^2$, $\sigma_c^2$, $p$, and the threshold $\epsilon$.}
     \Output{${\bf{w}}^{\star}$,${\bf{A}}^{\star},{\bf{B}}^{\star}$, ${\bm{\Phi}}^{\star},{\bm{\Lambda}}^{\star}$.}

     Initialize $ ({\bf{w}}^{0},{{\mu}}^{0},{\bf{v}}^{0},{\bf{u}}^{0},{\bf{q}}^{0},{\bm{\gamma}}^{0},{\bf{a}}^{0},{\bf{z}}^{0})$, and set the initial penalty ${\rho}$.\\ 
      \Repeat{ Convergence }
      {

       Update the receive beamforming vector ${\bf{w}}^{t}$ by \eqref{receive beamforming}, \\

       Update the amplification factor $\mu^{t}$ by \eqref{mu opt}, \\

       Update the auxiliary variables $\{ {\bf{u}}^{t},{\bf{v}}^{t},{\bf{q}}^{t} \}$ by \eqref{aux var}, \\

       Update the selection variables $\{ {\bm{\gamma}}^{t},{\bf{a}}^{t},{\bf{z}}^{t} \}$ via the CVX solver by solving \eqref{mode sel}, \eqref{ant sel} and \eqref{dis phase},\\
       
       If $(t \mod T) = 0$ \\
       \ \ Increase the penalty $\rho^{t}$ , \\
       End
      }
      Recover ${\bf{A}}^{\star}$, ${\bf{B}}^{\star}$, ${\bm{\Phi}}^{\star}$, ${\bm{\Lambda}}^{\star}$ from ${{\mu}}^{\star}$, ${\bf{v}}^{\star}$, ${\bf{u}}^{\star}$, ${\bf{q}}^{\star}$, ${\bm{\gamma}}^{\star}$, ${\bf{a}}^{\star}$, ${\bf{z}}^{\star}$.        
\end{algorithm}

\vspace{-18pt}
\section{Simulation}
In this section, we conduct numerical simulations to evaluate the performance of the dynamic HRIS-aided communication system. Unless otherwise stated, we set the basic system parameters as follows: $N_R=32$, $L=8$, $N=64$, $B=2$, $k_t=k_r=0.08$ \cite{holma2011lte}. 
{\color{blue}Under a three-dimensional deployment setup, the BS and the HRIS are deployed at $(0,80,5)$m and $(50,50,15)$m, respectively.} The user is located at $(0,0,2)$m.
The transmit power for the user is $p=10$ dBm.
The noise powers are given by $\sigma_b^2=\sigma_c^2=\sigma^2=-80$dBm.
The channel ${\bf{h}}_{d}$ is assumed to be Rayleigh fading, while the HRIS-related channels ${\bf{h}}_{r}$ and ${\bf{G}}$ obey Rician fading with the Rician factor being 0.75.
For the large-scale fading, the path loss model $\text{PL}(d)\!=\!\beta_0 (d)^{-{\alpha}_{p}}$ is considered,  where $d$  and ${\alpha}_{p}$ respectively denote the propagation distance and the path loss exponent. $\beta_0$ stands for the path loss at the reference distance of 1m and is set as $-30$dB.  
In particular,  we set ${\alpha}_{p}^{RB}=2.2$,  ${\alpha}_{p}^{UR}=2.2$ and ${\alpha}_{p}^{UB}=3.5$ for the BS-HRIS channel ${\bf G}$, the HRIS-user channel ${\bf h}_{r}$ and the BS-user channel ${\bf h}_{d}$, respectively.  
The conventional passive RIS and active RIS schemes are adopted as the benchmarks. Besides, the hybrid RIS with a fixed number of active elements, denoted as ``F-HRIS'', is compared to illustrate the superiority of the proposed dynamic HRIS scheme (``D-HRIS''). The non-robust scheme of ignoring the impact of HWIs, represented as ``N-HRIS'', is also considered. {\color{blue}``D-HRIS w/o AS'' denotes the dynamic HRIS-aided MISO communication system with fixed BS antenna selection.}
For a fair comparison, we assume that the total power budgets for different schemes are the same.

    \begin{figure}[t]
    \centering  
     \includegraphics[width=0.35\textwidth]{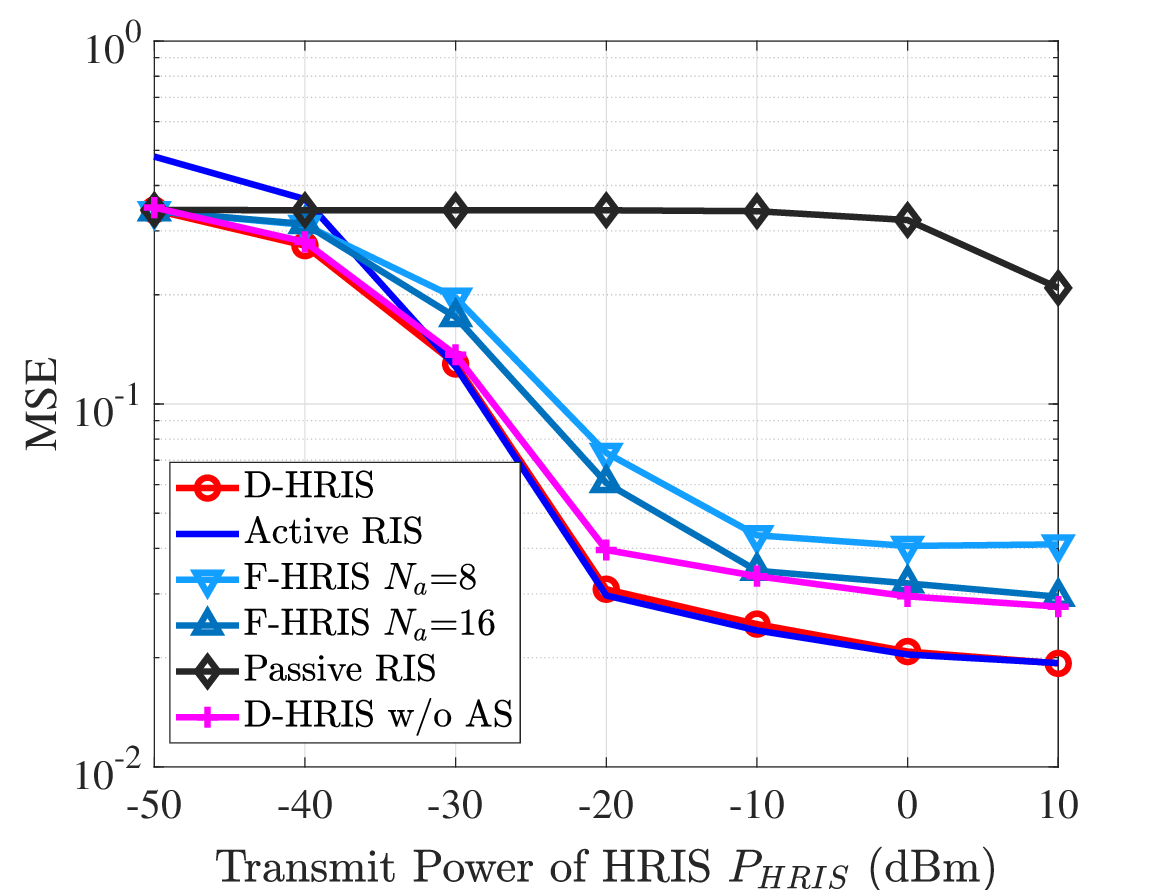}
     \vspace{-8pt}
        \caption{{\color{blue}MSE versus the transmit power of HRIS for different algorithm comparisons.}}
        \vspace{-14pt}
        \label{Algorithm comparison}    
    \end{figure}

    \begin{figure}[t]
    \centering  
     \includegraphics[width=0.35\textwidth]{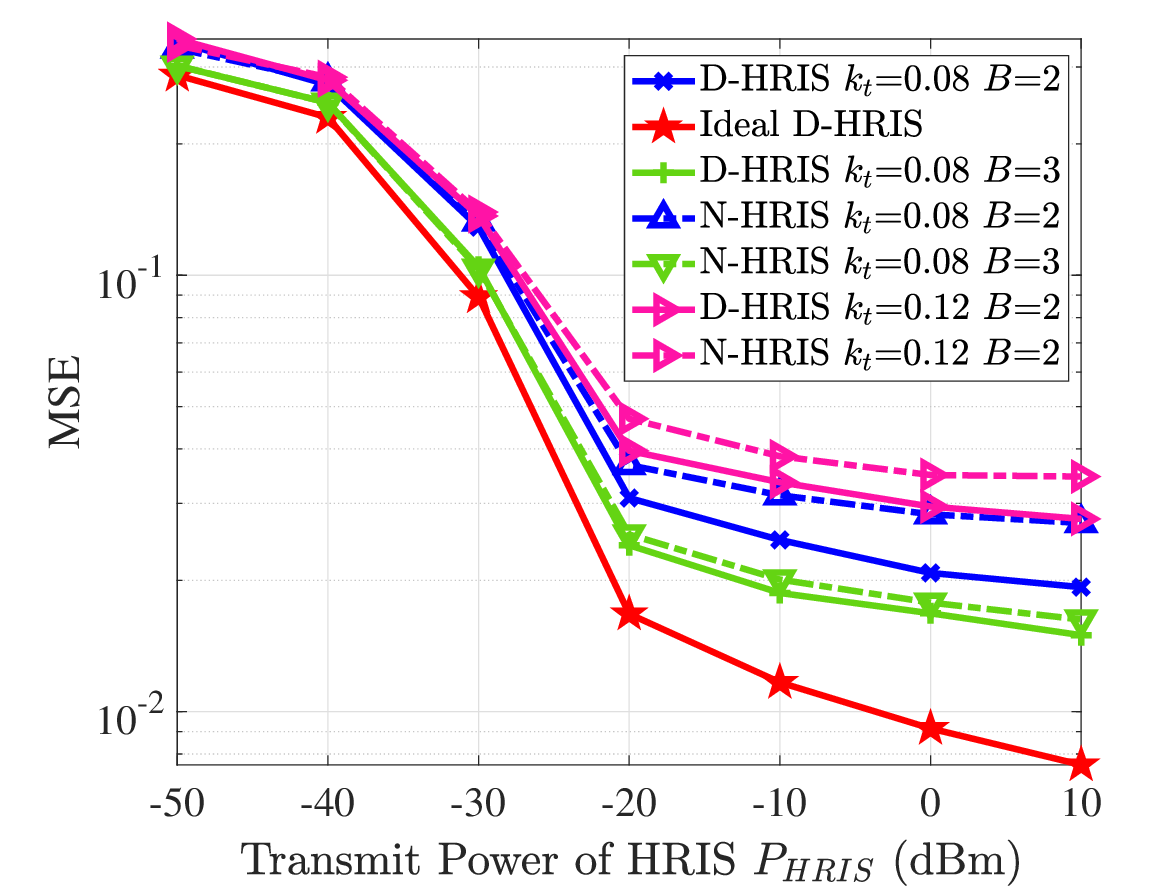}
     \vspace{-8pt}
        \caption{Robust versus nonrobust design under the transmit power of HRIS.}
        \vspace{-14pt}
        \label{Nonrobust}    
    \end{figure}

{\color{blue}  The MSE performance versus the transmit power of HRIS $P_{\rm HRIS}$ for different schemes are compared in Fig.~\ref{Algorithm comparison}. 
In Fig.~\ref{Algorithm comparison}, the performance of the passive RIS scheme improves with an increase in the transmit power of the HRIS, denoted as 
$P_{HRIS}$. This improvement is due to the corresponding increase in the transmit power at the BS for passive RIS to ensure a fair comparison. 
The proposed D-HRIS achieves the same performance as the active RIS in the high transmit power region. In contrast, the D-HRIS performs better than the active RIS in the lower-power regime. It happens because the transmit power is not sufficient to activate all the elements in the active RIS. The performance gap between the D-HRIS and F-HRIS demonstrates the effectiveness of the dynamic configuration for the active elements. Besides, the MSE performance of passive RIS also increases as the transmit power of HRIS increases, since the total power budget increases correspondingly. The performance gap between the ``D-HRIS'' and ``D-HRIS w/o AS'' schemes demonstrates the necessity of BS antenna selection, as it can effectively combat multipath fading by adjusting the antennas.}
In Fig.~\ref{Nonrobust}, the ideal D-HRIS limits the MSE performance of the HWI-aware systems to an upper bound. 
The D-HRIS always achieves better performance than the N-HRIS scheme under various setups of HWIs. 
It showcases the effectiveness of the robust design under imperfect hardware implementations.

\vspace{-6pt}
\section{Conclusion}

In this paper, we investigated the channel-aware placement of receive antennas at the BS and active/passive elements at the HRIS under the imperfect hardware implementation. The average MSE minimization problem for the dynamic HRIS-aided MISO system was studied by jointly optimizing the receive antenna selection matrix, the reflection phase coefficients, the reflection amplitude matrix, and the mode selection matrix for the active elements. To tackle the non-convexity and intractability of this problem, we proposed a PEBCD algorithm to solve these variables alternately. Numerical simulations showed great superiority of the channel-aware configuration for the receive antennas and active/passive elements in the HRIS-aided MISO system as compared to the conventional benchmark schemes.

\vspace{-6pt}
\bibliographystyle{IEEEtran}
\bibliography{HWIs_Ant_RIS}

\vfill

\end{document}